\IfFileExists{IEEEtran.cls}{%
  \documentclass[journal]{IEEEtran}
}{%
  \documentclass[10pt,twocolumn]{article}
  \usepackage[letterpaper,margin=0.67in,columnsep=0.24in]{geometry}
}
\providecommand{\IEEEPARstart}[2]{#1#2}

\usepackage{amsmath,amssymb,amsthm,mathtools}
\usepackage[T1]{fontenc}
\IfFileExists{IEEEtran.cls}{}{\usepackage{lmodern}}
\usepackage{xcolor}
\usepackage{balance}
\usepackage{float}
\usepackage{cite}
\usepackage{booktabs}
\usepackage{microtype}
\usepackage[hidelinks]{hyperref}

\newtheorem{theorem}{Theorem}
\newtheorem{proposition}[theorem]{Proposition}
\newtheorem{lemma}[theorem]{Lemma}
\newtheorem{corollary}[theorem]{Corollary}
\newtheorem{remark}{Remark}

\newcommand{\A}{\mathcal{A}}
\newcommand{\C}{C}

\newcommand{\D}{D}
\newcommand{\E}{\mathbb{E}}
\newcommand{\Z}{\mathbb{Z}}
\newcommand{\N}{\mathbb{Z}_{\geq 1}}
\newcommand{\Prb}{\mathbb{P}}
\newcommand{\supp}{\operatorname{supp}}
\newcommand{\myqed}{\hfill\(\blacktriangle\)}

\makeatletter
\@ifundefined{IEEEkeywords}{%
  \newenvironment{IEEEkeywords}
    {\par\smallskip\noindent\textbf{Index Terms---}}
    {\par\smallskip}
}{}
\makeatother

\title{The Capacity of a Family of Sticky Channels}

\author{Mladen~Kova\v{c}evi\'c%
\thanks{M.~Kova\v{c}evi\'c is with the Faculty of Technical Sciences,
University of Novi Sad, 21000 Novi Sad, Serbia
(e-mail: kmladen@uns.ac.rs).}
\thanks{This work was supported by the Ministry of Science, Technological Development
and Innovation of the Republic of Serbia (contract no. 451-03-34/2026-03/200156)
and by the Faculty of Technical Sciences, University of Novi Sad, Serbia (project
no. 01-3609/1).}}

\date{}

\hypersetup{
  pdftitle={The Capacity of a Family of Sticky Channels},
  pdfauthor={Mladen Kovacevic},
  pdfsubject={Exact capacities, domination, and compound Fuss-Catalan sticky channels},
  pdfkeywords={sticky channel, repeat channel, Shannon capacity,
  zero-error capacity, Fuss-Catalan numbers, compound distribution,
  capacity per unit cost}
}

\begin{document}

\maketitle

\begin{abstract}
\boldmath{
We determine the capacity of a family of \(q\)-ary
sticky-insertion channels. Fix \(q\geq2\) and \(d\geq1\), and let
\(\lambda\) be the unique positive solution of \(\lambda^d = (q-1) (\lambda^{d-1} + \cdots + \lambda + 1 )\). We prove that, for every
repetition law supported on \(1+d\mathbb Z_{\geq0}\) and satisfying a
coefficientwise-domination criterion with domination constant \(\gamma\geq\lambda^{-d}\), the Shannon capacity equals the zero-error capacity, both being
\(\log_2\lambda\) bits per symbol. We also exhibit explicit repetition laws satisfying these conditions, one of which is given by the weighted Fuss--Catalan numbers.  To the best of our knowledge, these are the first known cases of nontrivial repeat channels whose Shannon capacity has
been determined exactly.
}
\end{abstract}


\section{Introduction}

\IEEEPARstart{I}{nsertions}, deletions, and duplications, collectively known as synchronization errors, arise in
a range of communication and data storage systems and constitute a classical
challenge in information theory \cite{Mitzenmacher2009,CheraghchiRibeiro2021}.
Despite the simplicity of their definitions, the fundamental limits of the corresponding channels remain elusive, with exact capacity formulas being
exceptionally rare.
Dobrushin's coding theorem established an information-capacity
characterization under certain regularity conditions \cite{Dobrushin1967}, but even the binary deletion channel has resisted an
exact evaluation.  

Repeat channels are special cases of channels with synchronization errors in which each input symbol is replaced by a random number
of identical copies.  When the number of copies is positive almost surely,
the channel has no deletions and is called a \emph{sticky-insertion channel} or simply a \emph{sticky channel}.  Such a channel preserves the input run
symbols and randomizes only the run lengths.  Mitzenmacher \cite{Mitzenmacher2008} used this
observation to express the capacity of a sticky channel as the capacity per unit cost of
the corresponding run-length channel.  Numerical and analytical bounds on the capacity were developed in
\cite{DrineaMitzenmacher2007,IyengarSiegelWolf2016,KirschDrinea2010,
MercierTarokhLabeau2012,Mitzenmacher2008,RamezaniArdakani2013}.  Cheraghchi's
general KL-duality method \cite{Cheraghchi2019} and its specialization to sticky
channels by Cheraghchi and Ribeiro \cite{CheraghchiRibeiro2019} produced
sharp analytical upper bounds.
Recent work also gives efficient codes for
repeat channels with square-integrable laws
\cite{PerniceLiWootters2022} and improved bounds for the Poisson-repeat
channel \cite{KazemiDuman2023}.  For the standard random laws
considered in the literature, the exact capacity generally remains
unknown.
The zero-error capacity, however, has been determined for a broad class of sticky channels \cite{Kovacevic2019}, as well as for closely related tandem-duplication channels \cite{JainFarnoudSchwartzBruck2017,Kovacevic2019,Kovacevic2022}, which are of interest in DNA-based data storage systems.

In this paper, we present a class of sticky-insertion channel laws for which the capacity can be determined exactly.
Fix an
alphabet size \(q\geq2\), an integer \(d\geq1\), and let \(\lambda=\lambda_{q,d}\) be the unique positive
solution of
\begin{equation}
    \lambda^d = (q-1) (\lambda^{d-1} + \cdots + \lambda + 1 ) .
\label{eq:lambda-polynomial-intro}
\end{equation}
Denote by \(K\) the random number of copies an input symbol is replaced with in the channel, and by \(W_m(y)\) the probability that an input run of length \(m\) produces an output run of length \(y\). We prove that, if
\begin{equation}
    \supp(K)\subseteq1+d\mathbb Z_{\geq0}
\label{eq:arithmetic-support-intro}
\end{equation}
and the distributions \(W_m\) satisfy
\begin{equation}
    W_m(y)\geq\gamma W_{m+d}(y)
    \qquad \forall m,y\geq 1
\label{eq:domination-intro}
\end{equation}
for some \(\gamma\geq\lambda^{-d}\), then
\begin{equation}
  \C_q=\C_{0,q}=\log_2\lambda .
\label{eq:main-intro}
\end{equation}
An explicit
dual distribution demonstrates that, under the stated conditions, the vanishing-error capacity cannot exceed the rate of the
modulo-\(d\) zero-error construction.
Moreover, this construction is also explicit and yields maximum-cardinality zero-error
codes admitting efficient enumerative encoding and decoding.
We note that the condition \eqref{eq:arithmetic-support-intro} alone does not force the result; it provides the
zero-error code, but the output distributions of run lengths \(m\) and
\(m+d\) still carry different soft information.

We also characterize the class of laws satisfying \eqref{eq:arithmetic-support-intro} and \eqref{eq:domination-intro} and provide explicit examples. One example is the power law
\(\Prb(K=1+dk)\propto(k+1)^{-\alpha}\).
Another example is given by the following probability mass function, representing weighted Fuss--Catalan numbers:
\begin{equation}
\Prb(K=dk+1)
=\frac{1}{dk+1}\binom{(d+1)k}{k}(1-\beta)^{dk+1}\beta^k
\label{eq:pmf-intro}
\end{equation}
for \(k\geq0\), where \(\beta\in(0,\frac{1}{d+1}]\) is a parameter. For the family \eqref{eq:pmf-intro} we obtain a sharp characterization stating that \eqref{eq:main-intro} holds if and \emph{only if} \(\beta\geq\lambda^{-d}\). (At the endpoint
\(\beta=\frac{1}{d+1}\) the result remains
operationally valid even though the law has infinite mean.) Below this threshold the Shannon capacity \(\C_q\) strictly exceeds the zero-error capacity \(\C_{0,q}=\log_2\lambda\).

The results, in particular, answer the question posed by Cheraghchi and Ribeiro \cite{CheraghchiRibeiro2019} of finding a nontrivial repeat channel whose capacity can be determined exactly through their relative-entropy duality framework.

\section{Operational Model and Run-Length Reduction}
\label{sec:model}

All logarithms are to base two. Let \(\A\) be an alphabet of size
\(q\geq2\).  Let \(K\) be a random variable taking values in
\(\N\).
For \(x=x_1\cdots x_n\in\A^n\), the channel
output is
\begin{equation}
    Y=x_1^{K_1}x_2^{K_2}\cdots x_n^{K_n},
    \label{eq:sticky-channel}
\end{equation}
where \(K_1,\ldots,K_n\) are independent copies of \(K\), and \(x_i^k\)
denotes \(k\) consecutive copies of \(x_i\).

An \((n,N,\epsilon)\) code has \(N\) codewords in \(\A^n\) and average
decoding error at most \(\epsilon\) under equiprobable messages.  A rate
\(R\) is achievable if there is a sequence of such codes, indexed by \(j\),
for which \(n_j\to\infty\), \(\epsilon_j\to0\), and
\begin{equation}
    \liminf_{j\to\infty}\frac{1}{n_j}\log_2N_j\geq R.
    \label{eq:achievable-rate}
\end{equation}
The Shannon capacity \(\C_q\) is the supremum of achievable rates.

For \(x\in\A^n\), let \(\Gamma_n(x)\) be the support of the output
distribution induced by \(x\).  A code is zero-error if
\(\Gamma_n(x)\cap\Gamma_n(x')=\varnothing\) for every two distinct
codewords \(x, x'\).  If \(M_0(n)\) is the maximum size of such a code, the
zero-error capacity \cite{Shannon1956} is the
quantity
\begin{equation}
    \C_{0,q}=\limsup_{n\to\infty}\frac{1}{n}\log_2M_0(n).
    \label{eq:zero-error-definition}
\end{equation}
Clearly \(\C_{0,q}\leq\C_q\).

For an input run of length \(\ell\), define
\begin{equation}
    S_\ell=K_1+\cdots+K_\ell,
\label{eq:run-channel}
\end{equation}
let \(W_\ell\) be the distribution of \(S_\ell\),
\begin{equation}
    W_\ell(y)=\Prb(S_\ell=y) ,
\label{eq:run-channel2}
\end{equation}
and write
\begin{equation}
    G(z)=\E z^K,
    \qquad
    W_\ell(y)=[z^y]G(z)^\ell .
    \label{eq:run-pgf}
\end{equation}

The DMC \(\ell\mapsto W_\ell\), with input cost \(\ell\), is the
\emph{run-length channel}.  The following is a self-contained \(q\)-ary
form of the run-length reduction in \cite{Mitzenmacher2008}, using the
capacity-per-unit-cost framework of \cite{Verdu1990,AbdelGhaffar1993}.
We include a proof because we require the
\(q\)-ary version under the fixed-input-length operational definition,
without any moment assumption (such as finite mean) on the repetition law.

\begin{theorem}[Operational run-length reduction]
\label{thm:run-reduction}
For every \(q\geq2\) and every repetition law on \(\N\),
\begin{equation}
 \C_q
 =
 \sup_{\substack{P_L\\\mathrm{finite\ support}}}
 \frac{I(L;S_L)+\log_2(q-1)}{\E L}.
 \label{eq:qary-capacity-per-cost}
\end{equation}
\end{theorem}

\begin{IEEEproof}
Denote the right-hand side of \eqref{eq:qary-capacity-per-cost} by
\(c^\star\).  For achievability, first take a rational, finitely supported
distribution \(P_L\), and fix \(\delta>0\).  Since the run input alphabet
is finite, there is a finite output quantization \(\pi\) such that
\(I(L;\pi(S_L))\geq I(L;S_L)-\delta\).  For \(m\) tending to infinity
through multiples of the denominator of \(P_L\), the finite-output
constant-composition theorem \cite[Ch.~6]{CsiszarKorner2011} gives
length-\(m\) codes of type \(P_L\),
error tending to zero, and logarithmic size
\begin{equation}
    m\bigl(I(L;S_L)-\delta-o(1)\bigr).
\end{equation}
Every codeword has the same total input cost \(n=m\E L\).  Letting
\(\delta\downarrow0\) gives the desired run-length rate.

Independently choose any valid run-symbol sequence
\((A_1,\ldots,A_m)\), where adjacent symbols differ.  There are
\(q (q-1)^{m-1}\) such sequences, and the sticky channel reveals the chosen
sequence without error.  Combining the two codes gives rate
\begin{equation}
 \frac{I(L;S_L)-\delta+\log_2(q-1)-o(1)}{\E L}
 +\frac{\log_2(\frac{q}{q-1})}{m\E L}.
\end{equation}
For a fixed finite input support, mutual information is continuous in the
input distribution even when the output alphabet is countable: varying
\(P_L\) changes the joint law of \((L,S_L)\) in total variation, and
conditional entropy is continuous under total variation when \(L\) has a
fixed finite alphabet \cite[Ch.~3]{CsiszarKorner2011}.
Rational finite-support laws therefore approximate every rate in the
supremum, which proves \(\C_q\geq c^\star\).

For the converse, an arbitrary \(X^n\) has a unique run representation
\begin{equation}
    X^n\longleftrightarrow(M,A^M,L^M),
    \qquad \sum_{i=1}^{M}L_i=n.
\end{equation}
Because \(K\geq1\), the output representation is
\begin{equation}
    Y\longleftrightarrow(M,A^M,S^M);
\end{equation}
in particular, \(M\) and \(A^M\) are revealed exactly.  Thus
\begin{equation}
 I(X^n;Y)
 =
 H(M,A^M)+I(L^M;S^M\mid M,A^M).
\label{eq:information-decomposition}
\end{equation}
%
Condition on \(M=m\) and \(A^M=a^m\), and let
\begin{equation}
    Q_i=P_{S_i\mid m,a^m},
    \qquad
    Q^{(m)}=\bigotimes_{i=1}^m Q_i.
\end{equation}
Since the channel acts independently on the runs,
\begin{equation}
    P_{S^m\mid L^m,m,a^m}
    =\bigotimes_{i=1}^m W_{L_i}.
\end{equation}
The relative-entropy decomposition of mutual information therefore gives
\begin{align}
\nonumber
 I(L^m;S^m\mid m,a^m)
 &=
 \E\!\left[
 D\!\left(
 \bigotimes_{i=1}^m W_{L_i}
 \,\middle\Vert\,Q^{(m)}
 \right)
 \middle|m,a^m\right]\\
\nonumber
 &\quad{}
 -D\!\left(
 P_{S^m\mid m,a^m}
 \,\middle\Vert\,Q^{(m)}
 \right)\\
\nonumber
 &\leq
 \sum_{i=1}^m
 \E\!\left[
 D(W_{L_i}\Vert Q_i)
 \middle|m,a^m\right]\\
 &=
 \sum_{i=1}^m I(L_i;S_i\mid m,a^m).
\end{align}
This argument uses only relative entropy and therefore does not require
the output entropies to be finite.

Every conditional marginal \(L_i\) is supported on \(\{1,\ldots,n\}\).
The definition of \(c^\star\) therefore gives
\begin{equation}
 I(L^m;S^m\mid m,a^m)
 \leq c^\star n-m\log_2(q-1).
\end{equation}
Moreover,
\begin{equation}
 H(M,A^M)
 \leq\log_2n+\log_2q+(\E M-1)\log_2(q-1).
\label{eq:HM}
\end{equation}
Substitution into \eqref{eq:information-decomposition} gives the
fixed-block bound
\begin{equation}
 I(X^n;Y)
 \leq c^\star n+\log_2n+\log_2\!\Big(\frac{q}{q-1}\Big).
 \label{eq:fixed-block-capacity-bound}
\end{equation}
For an \((n,N,\epsilon)\) code, Fano's inequality and data processing give
\begin{equation}
 \frac{1}{n}\log_2N
 \leq
 \frac{c^\star}{1-\epsilon}
 +\frac{\log_2n+\log_2(\frac{q}{q-1})+h_2(\epsilon)}
 {n(1-\epsilon)}.
 \label{eq:finite-block-fano}
\end{equation}
Letting \(n\to\infty\) and \(\epsilon\to0\) proves
\(\C_q\leq c^\star\).
\end{IEEEproof}

The first run symbol contributes only the vanishing constant
\(\log_2 q\); each subsequent run contributes \(\log_2(q-1)\) noiseless
bits.  For \(q=2\), this term is zero and
\eqref{eq:qary-capacity-per-cost} reduces to the familiar binary formula
\cite{Mitzenmacher2008}.

We will use the following dual consequence in a form that already includes
the run-symbol contribution.  It is the KL-duality method of
\cite{Cheraghchi2019}, specialized to the run-length channel as in
\cite{CheraghchiRibeiro2019}; the contribution below is to identify an
explicit reference distribution for which the bound is exactly tight.

\begin{lemma}[Run-length dual bound]
\label{lem:qary-dual}
Suppose a probability distribution \(Q\) on \(\N\) and a number \(c\)
satisfy
\begin{equation}
    \D(W_\ell\Vert Q)
    \leq c\ell-\log_2(q-1)
    \qquad\forall\ell\geq1.
    \label{eq:qary-dual-condition}
\end{equation}
Then \(\C_q\leq c\).
\end{lemma}

\begin{IEEEproof}
For fixed \(M=m,A^M=a^m\), use \(Q^{\otimes m}\) as a reference output
distribution, and write
\(W_{L^m}=\bigotimes_{i=1}^{m}W_{L_i}\).  The relative-entropy identity
gives
\begin{align}
\nonumber
 I(L^m;S^m\mid m,a^m)
 &=
 \E\!\left[
 \D(W_{L^m}\Vert Q^{\otimes m})
 \,\middle|\,m,a^m\right]\\
\nonumber
 &\quad{}-\D(P_{S^m\mid m,a^m}\Vert Q^{\otimes m})\\
\nonumber
 &\leq
 \E\!\left[
 \sum_{i=1}^{m}\D(W_{L_i}\Vert Q)
 \,\middle|\,m,a^m\right]\\
 &\leq cn-m\log_2(q-1).
\end{align}
The first expectation is finite because each \(L_i\leq n\) and every
\(\D(W_\ell\Vert Q)\) has the finite bound
\eqref{eq:qary-dual-condition}; convexity then makes the subtracted
divergence finite as well.  Using this estimate and the identities \eqref{eq:information-decomposition} and \eqref{eq:HM} yields
\begin{equation}
    I(X^n;Y)
    \leq cn+\log_2n+\log_2\!\Big(\frac{q}{q-1}\Big).
\end{equation}
For an \((n,N,\epsilon)\) code, Fano's inequality and data processing give
\begin{equation}
    (1-\epsilon)\log_2N
    \leq I(X^n;Y)+h_2(\epsilon).
\end{equation}
Dividing by \(n\), and then letting \(n\to\infty\) and
\(\epsilon\to0\), completes the proof.
\end{IEEEproof}

\section{Zero-Error Capacity and a General Domination Criterion}
\label{sec:criterion}

Fix \(q\geq2\) and \(d\geq1\), and let \(\lambda=\lambda_{q,d}\geq1\)
be defined by
\begin{equation}
    (q-1)\sum_{r=1}^{d}\lambda^{-r}=1.
    \label{eq:lambda-definition}
\end{equation}
The left-hand side is strictly decreasing in \(\lambda>0\), so the solution
is unique.  The equality \(\lambda=1\) occurs only when \((q,d)=(2,1)\).
The first term in \eqref{eq:lambda-definition} also shows that
\(\lambda\geq q-1\).

\begin{theorem}[Coefficientwise-domination criterion]
\label{thm:domination}
Suppose
\begin{equation}
    \supp(K)\subseteq1+d\mathbb Z_{\geq0}
    \label{eq:arithmetic-support}
\end{equation}
and, for some \(\gamma\in(0,1]\),
\begin{equation}
    W_m(y)\geq\gamma W_{m+d}(y)
    \qquad\forall m,y\geq1 .
    \label{eq:general-domination}
\end{equation}
Then
\begin{align}
\label{eq:ze}
 &\log_2\lambda
 = \C_{0,q}  \\
 &\leq\C_q
 \leq
 \min\left\{\log_2q,\,
 \max\left\{\log_2\lambda,\,
 \frac{1}{d}\log_2\frac{1}{\gamma}\right\}\right\}.
\label{eq:general-bound}
\end{align}
In particular, if \(\gamma\geq\lambda^{-d}\), then
\begin{equation}
    \C_q=\C_{0,q}=\log_2\lambda.
\label{eq:general-exact}
\end{equation}
\end{theorem}

\begin{IEEEproof}
We first determine the zero-error capacity. The argument is an adaptation of that from \cite{Kovacevic2019} to the present setting; we nonetheless give the full proof for completeness.

Write an input word in run form as
\begin{equation}
    x=a_1^{\ell_1}\cdots a_m^{\ell_m},
    \qquad a_i\neq a_{i+1}.
\end{equation}
Since every repetition count is positive, the sticky channel preserves the
run-symbol sequence \((a_1,\ldots,a_m)\).  Moreover,
\eqref{eq:arithmetic-support} implies
\begin{equation}
    \supp(W_\ell)
    \subseteq
    \ell+d\mathbb Z_{\geq0}.
\end{equation}
Thus output-length distributions corresponding to input run lengths in
different residue classes modulo \(d\) have disjoint supports.

On the other hand, since \(\gamma>0\),
\eqref{eq:general-domination} implies
\begin{equation}
    \supp(W_{r+d})\subseteq\supp(W_r)
    \qquad\forall r\geq1.
\end{equation}
Iterating this inclusion gives
\begin{equation}
    \supp(W_{r+jd})\subseteq\supp(W_r)
    \qquad \forall r\geq1,\ j\geq0.
\end{equation}
Consequently, two input runs (of the same symbol) are confusable if and only if their lengths are
congruent modulo \(d\). Indeed, as noted in the previous paragraph, different residues give disjoint output
supports.  Conversely, if \(\ell'\geq\ell\) and
\(\ell'\equiv\ell\pmod d\), then
\(\supp(W_{\ell'})\subseteq\supp(W_\ell)\), so the two distributions have
a nonempty common support.

It follows that two input words
\begin{equation}
    x=a_1^{\ell_1}\cdots a_m^{\ell_m}
    \quad\text{and}\quad
    x'={a'_1}^{\ell'_1}\cdots {a'_{m'}}^{\ell'_{m'}}
\end{equation}
are confusable if and only if
\begin{equation}
    m=m'\quad\text{and}\quad
    a_i=a'_i,\quad
    \ell_i\equiv\ell'_i\pmod d
    \quad\forall i.
\end{equation}
To verify the converse explicitly, suppose these conditions hold and put
\begin{equation}
    t_i=\max\{\ell_i,\ell'_i\}.
\end{equation}
The support nesting above implies
\begin{equation}
    \supp(W_{t_i})
    \subseteq
    \supp(W_{\ell_i})\cap\supp(W_{\ell'_i}).
\end{equation}
Choose any \(y_i\in\supp(W_{t_i})\).  Then the output word
\begin{equation}
    a_1^{y_1}\cdots a_m^{y_m}
\end{equation}
has positive probability under both inputs, because the channel acts
independently on the runs.

For each run length \(\ell_i\), let \(r_i\in\{1,\ldots,d\}\) be its unique
representative modulo \(d\), with \(r_i=d\) representing residue zero.  The
confusability class of \(x\) is therefore completely specified by the
signature
\begin{equation}
    \sigma(x)
    =
    \bigl((a_1,\ldots,a_m),(r_1,\ldots,r_m)\bigr).
\end{equation}
In particular, confusability is an equivalence relation.  A zero-error code
can contain at most one word from each signature realized at length \(n\),
while choosing one representative from every realized signature produces a
zero-error code.  Hence \(M_0(n)\) equals the number of signatures realized
by words of length \(n\).
Estimating this number follows a standard combinatorial argument, which is given in the Appendix, implying
\begin{equation}
\label{eq:M0}
    \limsup_{n\to\infty}
    \frac{1}{n}\log_2M_0(n)
    =
    \log_2\lambda.
\end{equation}
This proves \eqref{eq:ze}.

Let us now derive the upper bound on the Shannon capacity stated in \eqref{eq:general-bound}.
Set
\begin{equation}
\begin{split}
 c&=\max\left\{\log_2\lambda,\,
 \frac{1}{d}\log_2\frac{1}{\gamma}\right\},\\
 t&=2^{-c},\qquad
 Z=(q-1)\sum_{r=1}^{d}t^r.
\end{split}
\label{eq:general-dual-parameters}
\end{equation}
Then \(Z\leq1\) and \(t^d\leq\gamma\).  Define
\begin{equation}
    p_r=\frac{(q-1) t^r}{Z},\quad 1\leq r\leq d,
    \qquad Q=\sum_{r=1}^{d}p_rW_r.
    \label{eq:general-dual}
\end{equation}
The distributions \(W_1,\ldots,W_d\) occupy distinct residue classes
modulo \(d\).  Write \(\ell=r+jd\), where \(1\leq r\leq d\) and
\(j\geq0\).  Iterating \eqref{eq:general-domination} gives
\(W_r\geq\gamma^jW_\ell\).  On the support of \(W_\ell\), only the
\(r\)-th component of \(Q\) contributes, so \(Q(y)=p_rW_r(y)\).
Hence
\begin{equation}
 \frac{W_\ell(y)}{Q(y)}
 \leq\frac{Z}{q-1}t^{-r}\gamma^{-j}
 \leq\frac{1}{q-1}t^{-(r+jd)}
 =\frac{1}{q-1}2^{c\ell}.
\label{eq:general-likelihood-bound}
\end{equation}
Taking logarithms and averaging gives
\begin{equation}
    \D(W_\ell\Vert Q)\leq c\ell-\log_2(q-1).
\end{equation}
Lemma~\ref{lem:qary-dual} then implies \(\C_q\leq c\), and the trivial input
alphabet bound gives \(\C_q\leq\log_2q\).  This proves
\eqref{eq:general-bound}. 

If \(\gamma\geq\lambda^{-d}\), then
\(c=\log_2\lambda\), matching the zero-error lower bound and proving \eqref{eq:general-exact}.
\end{IEEEproof}

The candidate input and reference output distributions at rate
\(\log_2\lambda\) are
\begin{subequations}
\label{eq:exact-dual}
\begin{align}
        P_0(L=r)&=p_r=(q-1)\lambda^{-r},\\
        Q_0&=(q-1)\sum_{r=1}^{d}\lambda^{-r}W_r.
\end{align}
\end{subequations}
The supports of \(W_1,\ldots,W_d\) are disjoint, so
\(I_{P_0}(L;S_L)=H(P_0)\).  Moreover,
\begin{align}
 I_{P_0}(L;S_L)+\log_2(q-1)
 &=H(P_0)+\log_2(q-1)\nonumber\\
 &=\E_{P_0} L\log_2\lambda.
\label{eq:optimal-run-entropy}
\end{align}
Thus the modulo-\(d\) code attains \(\log_2\lambda\) by using
run lengths \(1,\ldots,d\) with probabilities \(p_r\), while each next
run symbol carries \(\log_2(q-1)\) additional bits.

\begin{proposition}[Exact KL criterion]
\label{prop:exact-kl-criterion}
Under the hypotheses of Theorem~\ref{thm:domination}, let \(P_0\) and
\(Q_0\) be given by \eqref{eq:exact-dual}.  For
\(\ell=r+jd\), where \(1\leq r\leq d\) and \(j\geq0\), define
\begin{equation}
    \Delta_\ell
    \coloneq
    \D(W_\ell\Vert Q_0)
    -\ell\log_2\lambda
    +\log_2(q-1).
    \label{eq:delta-definition}
\end{equation}
Then
\begin{equation}
    \C_q=\C_{0,q}=\log_2\lambda
    \quad\Longleftrightarrow\quad
    \Delta_\ell\leq0
    \quad\forall\ell\geq1.
    \label{eq:exact-kl-characterization}
\end{equation}
Equivalently, the condition on the right is
\begin{equation}
    \D(W_{r+jd}\Vert W_r)
    \leq jd\log_2\lambda
    \qquad
    \forall\,1\leq r\leq d,\ j\geq1.
    \label{eq:within-residue-kl-condition}
\end{equation}
\end{proposition}

\begin{IEEEproof}
By the disjoint residue supports and the support nesting in
Theorem~\ref{thm:domination}, one has, on
\(\supp(W_{r+jd})\),
\begin{equation}
    Q_0(y)=(q-1)\lambda^{-r}W_r(y).
    \label{eq:exact-dual-on-residue}
\end{equation}
Consequently,
\begin{equation}
    \Delta_{r+jd}
    =
    \D(W_{r+jd}\Vert W_r)-jd\log_2\lambda,
    \label{eq:delta-within-residue}
\end{equation}
which proves the equivalence between the two displayed conditions.
Moreover, domination gives
\(W_{r+jd}/W_r\leq\gamma^{-j}\) on its support, so every divergence in
\eqref{eq:delta-within-residue}, and hence every \(\Delta_\ell\), is
finite.  In particular, \(\Delta_r=0\) for \(1\leq r\leq d\).

If \(\Delta_\ell\leq0\) for every \(\ell\), then
\eqref{eq:delta-definition} is precisely
\eqref{eq:qary-dual-condition} with \(Q=Q_0\) and
\(c=\log_2\lambda\).  Lemma~\ref{lem:qary-dual}, together with
Theorem~\ref{thm:domination}, therefore gives
\(\C_q=\C_{0,q}=\log_2\lambda\).

Conversely, suppose that \(\Delta_{\ell_0}>0\) for some \(\ell_0\).
Perturb the input-run law according to
\begin{equation}
    P_\epsilon=(1-\epsilon)P_0+\epsilon\delta_{\ell_0}
    \label{eq:general-run-perturbation}
\end{equation}
and define
\begin{equation}
    R_{\ell_0}(\epsilon)
    =
    \frac{I_{P_\epsilon}(L;S_L)+\log_2(q-1)}
         {\E_{P_\epsilon}L}.
    \label{eq:general-perturbed-rate}
\end{equation}
Let
\(Q_\epsilon=(1-\epsilon)Q_0+\epsilon W_{\ell_0}\) be the output law
induced by \(P_\epsilon\).  Applying the relative-entropy identity with
reference distribution \(Q_0\) gives
\begin{align}
 I_{P_\epsilon}(L;S_L)
 &=(1-\epsilon)I_{P_0}(L;S_L)
   +\epsilon\D(W_{\ell_0}\Vert Q_0)\nonumber\\
 &\quad{}-\D(Q_\epsilon\Vert Q_0).
 \label{eq:general-mutual-information-identity}
\end{align}
Splitting \(\log_2(q-1)\) into its \(1-\epsilon\) and \(\epsilon\)
parts and using \eqref{eq:optimal-run-entropy} yields the exact identity
\begin{align}
 R_{\ell_0}(\epsilon)-\log_2\lambda
 &=
 \frac{\epsilon\Delta_{\ell_0}-\D(Q_\epsilon\Vert Q_0)}
 {(1-\epsilon)\E_{P_0}L+\epsilon\ell_0}.
 \label{eq:general-rate-excess}
\end{align}
If \(\ell_0=r+jd\), the same domination bound and
\eqref{eq:exact-dual-on-residue} give the finite constant
\begin{equation}
 B\coloneq
 \sup_{y\in\supp(W_{\ell_0})}
 \frac{W_{\ell_0}(y)}{Q_0(y)}
 \leq
 \frac{\lambda^r}{(q-1)\gamma^j}.
 \label{eq:perturbation-likelihood-bound}
\end{equation}
Since \(Q_\epsilon-Q_0=\epsilon(W_{\ell_0}-Q_0)\),
\begin{align}
 \D(Q_\epsilon\Vert Q_0)
 &\leq(\log_2 e)\chi^2(Q_\epsilon\Vert Q_0)\nonumber\\
 &=(\log_2 e)\epsilon^2
   \chi^2(W_{\ell_0}\Vert Q_0)\nonumber\\
 &\leq(\log_2 e)(B-1)\epsilon^2,
 \label{eq:general-mixture-divergence}
\end{align}
where the last inequality follows from
\(\chi^2(W_{\ell_0}\Vert Q_0)
=\sum_y W_{\ell_0}(y)^2/Q_0(y)-1\leq B-1\).
It follows from \eqref{eq:general-rate-excess} that
\(R_{\ell_0}(\epsilon)>\log_2\lambda\) for all sufficiently small
\(\epsilon>0\).  Theorem~\ref{thm:run-reduction} therefore implies
\(\C_q>\log_2\lambda\), proving the converse.
\end{IEEEproof}

The KL condition \eqref{eq:within-residue-kl-condition} is an exact
optimality test, whereas coefficientwise domination is a checkable
pointwise sufficient condition.  Indeed,
\(W_r\geq\gamma^jW_{r+jd}\) implies
\begin{equation}
    \D(W_{r+jd}\Vert W_r)
    \leq j\log_2\frac{1}{\gamma}.
    \label{eq:pointwise-implies-kl}
\end{equation}
Thus \(\gamma\geq\lambda^{-d}\) implies
\eqref{eq:within-residue-kl-condition}.

\begin{remark}[Soft information]
\textnormal{
The distributions \(W_m\) and \(W_{m+d}\) are distinct, so output
lengths generally reveal more than a residue modulo \(d\).  Condition
\eqref{eq:general-domination} bounds the pointwise likelihood advantage from
adding \(d\) input symbols by \(1/\gamma\).  When
\(1/\gamma\leq\lambda^d\), that advantage cannot compensate for the cost of
the extra symbols at rate \(\log_2\lambda\).
\myqed}
\end{remark}

\begin{remark}[Domination forces unbounded support]
\textnormal{
Condition \eqref{eq:general-domination} cannot hold for a finitely
supported repetition law.  Indeed, if
\(M=\max\operatorname{supp}(K)\), then
\begin{equation}
    (d+1)M\in\operatorname{supp}(W_{d+1}).
\end{equation}
On the other hand, \eqref{eq:general-domination} with \(m=1\) and
\(\gamma>0\) implies
\begin{equation}
    \operatorname{supp}(W_{d+1})
    \subseteq\operatorname{supp}(W_1)
    =\operatorname{supp}(K),
\end{equation}
which would require \((d+1)M\leq M\), a contradiction.  Thus every
repetition law covered by Theorem~\ref{thm:domination} has unbounded
support.  In particular, arithmetic support alone does not imply
coefficientwise domination: any finitely supported law on
\(1+d\mathbb Z_{\geq0}\) violates it.
\myqed}
\end{remark}

\begin{remark}[Constructive capacity-achieving codes]
\textnormal{
The proof of Theorem~\ref{thm:domination} is constructive at every blocklength.  For each
signature
\(\sigma=((a_1,\ldots,a_m),(r_1,\ldots,r_m))\) realized at length
\(n\), let
\(j=(n-\sum_{i=1}^m r_i)/d\), and choose as its representative the
word
\(a_1^{r_1+jd}a_2^{r_2}\cdots a_m^{r_m}\).
The resulting code \(\mathcal C_n\) contains exactly one representative
from every confusability class and is therefore optimal,
\(\lvert\mathcal C_n\rvert=M_0(n)\).  Consequently, when
\(\gamma\geq\lambda^{-d}\), an appropriate sequence of these
maximum-cardinality zero-error codes achieves the Shannon capacity
\(\log_2\lambda\).
}

\textnormal{
These codes also admit efficient enumerative encoding and decoding.
The recurrence for \(T_k\) given in the Appendix permits the reduced
words, whose run lengths belong to \(\{1,\ldots,d\}\), to be ranked and
unranked by dynamic programming.  To decode a received word
\(a_1^{s_1}\cdots a_m^{s_m}\), one replaces each \(s_i\) by the unique
\(r_i\in\{1,\ldots,d\}\) congruent to \(s_i\) modulo \(d\), thereby
recovering the transmitted signature, and then applies the inverse
ranking procedure.  For fixed \(q\) and \(d\), the ranking and
unranking require \(O(n)\) arithmetic operations after \(O(n)\)
preprocessing, on integers having \(O(n)\) bits; parsing the received
word additionally requires time proportional to its length.
\myqed}
\end{remark}

\section{Fuss--Catalan and Compound Repetition Laws}
\label{sec:fuss-catalan}

\subsection{The Primitive Weighted Fuss--Catalan Law}

Fix \(d\geq1\) and \(\beta>0\), and let \(G=G_{d,\beta}\) be the formal
power-series solution of
\begin{equation}
    G(z)=(1-\beta)z+\beta G(z)^{d+1},
    \qquad G(0)=0.
    \label{eq:algebraic-pgf}
\end{equation}

In the proper range identified below, the coefficients generated by
\eqref{eq:algebraic-pgf} belong to a classical family of probability distributions.  More precisely,
after the affine reindexing
\begin{equation}
    J_m=\frac{S_m-m}{d},
\end{equation}
the law of \(J_m\) is the Jain--Consul generalized negative-binomial
distribution with parameters \(n=m\), \(b=d+1\), and \(p=\beta\)
\cite{JainConsul1971,ConsulGupta1980}.  This distribution is a standard
Lagrangian probability law
\cite{ConsulShenton1972,Good1975}; see also
\cite[Ch.~10]{ConsulFamoye2006}.  Equivalently, \(S_m\) is the number of
leaves in a Galton--Watson forest with \(m\) roots and offspring numbers
in \(\{0,d+1\}\).  We record the parameter conversion because the channel
uses the leaf count \(S_m\), rather than the usual generalized
negative-binomial variable \(J_m\).

\begin{proposition}[Repetition and run-length laws]
\label{prop:fuss-catalan-law}
The coefficients of \(G_{d,\beta}\) form a probability distribution if and
only if
\begin{equation}
    0<\beta\leq\frac{1}{d+1}.
    \label{eq:proper-range}
\end{equation}
In this range,
\begin{equation}
\begin{split}
 \Prb(K=dk+1)
 &=
 \frac{1}{dk+1}\binom{(d+1)k}{k}\\
 &\quad{}\times(1-\beta)^{dk+1}\beta^k,
 \qquad k\geq0,
\end{split}
\label{eq:fuss-catalan-pmf}
\end{equation}
with zero mass elsewhere.  Moreover, for every \(m\geq1\),
\begin{equation}
\begin{split}
 W_m(m+dk)
 &=
 \frac{m}{m+dk}
 \binom{m+(d+1)k-1}{k}\\
 &\quad{}\times(1-\beta)^{m+dk}\beta^k,
 \qquad k\geq0.
\end{split}
\label{eq:run-pmf}
\end{equation}
For \(0<\beta<\frac{1}{d+1}\),
\begin{equation}
    \E K=\frac{1-\beta}{1-(d+1)\beta}.
    \label{eq:mean-k}
\end{equation}
At \(\beta=\frac{1}{d+1}\), \(K\) is finite almost surely but
\(\E K=\infty\).
\end{proposition}

\begin{IEEEproof}
The coefficient of \(z\) in \(G\) is \(1-\beta\).  Thus \(\beta>1\)
cannot produce a probability distribution, while for \(\beta=1\) the
formal solution with \(G(0)=0\) is identically zero.  It remains to
consider \(0<\beta<1\).

Let \(\xi\) be an offspring variable satisfying
\begin{equation}
    \Prb(\xi=0)=1-\beta,
    \qquad
    \Prb(\xi=d+1)=\beta,
\end{equation}
and consider the corresponding ordered Galton--Watson tree.  If \(H(z)\)
is the possibly defective generating function of the number of leaves in
the finite tree, decomposition at the root gives
\begin{equation}
    H(z)=(1-\beta)z+\beta H(z)^{d+1},
    \qquad H(0)=0.
\end{equation}
The formal solution is unique, so \(H=G\).  Consequently, \(G(z)^m\)
is the generating function of the number of leaves in a forest with
\(m\) independent roots.

Dwass's formula \cite{Dwass1969} states that, if \(T_m\) is the
total number of vertices in such a forest, then
\begin{equation}
    \Prb(T_m=N)
    =
    \frac{m}{N}
    \Prb\left(\xi_1+\cdots+\xi_N=N-m\right).
\end{equation}
A forest with \(m\) roots and \(k\) internal vertices has
\(m+(d+1)k\) vertices and \(m+dk\) leaves.  Taking
\(N=m+(d+1)k\) in the preceding formula therefore gives
\begin{align}
\nonumber
 [z^{m+dk}]G(z)^m
 &=
 \frac{m}{m+(d+1)k}
 \binom{m+(d+1)k}{k}\\
\nonumber
 &\quad{}\times
 (1-\beta)^{m+dk}\beta^k\\
\nonumber
 &=
 \frac{m}{m+dk}
 \binom{m+(d+1)k-1}{k}\\
 &\quad{}\times
 (1-\beta)^{m+dk}\beta^k.
\end{align}
This is \eqref{eq:run-pmf}; setting \(m=1\) gives
\eqref{eq:fuss-catalan-pmf}.  In terms of
\(J_m=(S_m-m)/d\), the first expression is exactly the
Jain--Consul parameterization stated above.  The admissible parameter
range for that distribution was clarified in
\cite{ConsulGupta1980}.

The coefficients of \(G\) sum to the probability that the branching
process becomes extinct.  Its mean offspring number is
\(\E\xi=(d+1)\beta\), and the standard Galton--Watson extinction
criterion \cite{AthreyaNey1972} shows that extinction
occurs almost surely exactly when
\((d+1)\beta\leq1\).  This proves \eqref{eq:proper-range}.

Finally, differentiating \eqref{eq:algebraic-pgf} gives
\begin{equation}
    G'(z)
    =
    \frac{1-\beta}
    {1-(d+1)\beta G(z)^d}.
\end{equation}
In the subcritical case \(G(1)=1\), which yields
\eqref{eq:mean-k}.  At \(\beta=\frac{1}{d+1}\), extinction still occurs almost
surely, so \(K\) is finite almost surely, but \(G(z)\uparrow1\) as
\(z\uparrow1\) and hence \(G'(z)\to\infty\).  Therefore
\(\E K=\infty\) at criticality.
\end{IEEEproof}

The factor
\begin{equation}
    \frac{1}{dk+1}\binom{(d+1)k}{k}
\end{equation}
is the Fuss--Catalan number counting full ordered \((d+1)\)-ary trees
with \(k\) internal vertices and \(dk+1\) leaves
\cite{Raney1960,FlajoletSedgewick2009}.  Thus ``weighted
Fuss--Catalan'' describes the coefficients in
\eqref{eq:fuss-catalan-pmf}, while probabilistically the reindexed
variable \((K-1)/d\) is the Jain--Consul generalized negative-binomial
variable.

\subsection{Complete Characterization of the Domination Class}

We next show that the coefficientwise condition in
Theorem~\ref{thm:domination} has a complete algebraic description.  For
power series \(A\) and \(B\), write \(A\succeq B\) when every coefficient
of \(A-B\) is nonnegative.  To distinguish the primitive law from a
general repetition-count probability-generating function (PGF), set
\begin{equation}
    F_{d,\gamma}(u)\coloneq G_{d,\gamma}(u),
    \label{eq:primitive-notation}
\end{equation}
so that \(F_{d,\gamma}\) is the zero-constant-term formal solution of
\begin{equation}
    F_{d,\gamma}(u)
    =(1-\gamma)u+\gamma F_{d,\gamma}(u)^{d+1}.
    \label{eq:primitive-gamma-pgf}
\end{equation}

\begin{proposition}[Complete characterization of domination]
\label{prop:compound-characterization}
Let \(G\) be a probability-generating function supported on
\(1+d\mathbb Z_{\geq0}\), and fix \(\gamma\in(0,1]\).  The following
conditions are equivalent:
\begin{enumerate}
\item
\begin{equation}
    [z^y]G(z)^m
    \geq
    \gamma[z^y]G(z)^{m+d}
    \qquad\forall m,y\geq1;
    \label{eq:compound-condition-all-m}
\end{equation}
\item
\begin{equation}
    G(z)-\gamma G(z)^{d+1}\succeq0;
    \label{eq:compound-condition-m-one}
\end{equation}
\item one has \(\gamma\leq\frac{1}{d+1}\), and there is a unique
probability-generating function \(R\), supported on
\(1+d\mathbb Z_{\geq0}\), such that
\begin{equation}
    G(z)=F_{d,\gamma}(R(z)).
    \label{eq:compound-characterization}
\end{equation}
\end{enumerate}
Whenever these equivalent conditions hold, necessarily \(\gamma<1\), and
the PGF in the third condition is
\begin{equation}
    R(z)
    =
    \frac{G(z)-\gamma G(z)^{d+1}}{1-\gamma}.
    \label{eq:base-pgf-recovery}
\end{equation}
Moreover,
\begin{equation}
    G(z)
    =
    \sum_{k\geq0}
    \frac{1}{dk+1}\binom{(d+1)k}{k}
    \gamma^k(1-\gamma)^{dk+1} R(z)^{dk+1}.
\label{eq:compound-fuss-expansion}
\end{equation}
\end{proposition}

\begin{IEEEproof}
Condition \eqref{eq:compound-condition-all-m} with \(m=1\) is precisely
\eqref{eq:compound-condition-m-one}.  Conversely, the factorization
\begin{equation}
\begin{aligned}
    G(z)^m-\gamma G(z)^{m+d}
    &=
    G(z)^{m-1}
    \bigl(G(z)-\gamma G(z)^{d+1}\bigr)
\end{aligned}
    \label{eq:domination-factorization}
\end{equation}
shows that \eqref{eq:compound-condition-m-one} implies
\eqref{eq:compound-condition-all-m}, because \(G^{m-1}\) has nonnegative
coefficients.  Thus it is enough to test the domination inequality at
\(m=1\).

The value \(\gamma=1\) cannot satisfy
\eqref{eq:compound-condition-m-one}.  Indeed, \(G-G^{d+1}\) would have
nonnegative coefficients and coefficient sum zero, and hence
\begin{equation}
    G(z)=G(z)^{d+1}.
    \label{eq:gamma-one-impossible}
\end{equation}
This is impossible because the least positive exponent occurring in
\(G^{d+1}\) is \(d+1\) times the least positive exponent occurring in
\(G\).  We may therefore assume \(0<\gamma<1\) and define \(R\) by
\eqref{eq:base-pgf-recovery}.  Condition
\eqref{eq:compound-condition-m-one} makes its coefficients nonnegative,
and
\begin{equation}
    R(1)
    =
    \frac{G(1)-\gamma G(1)^{d+1}}{1-\gamma}
    =1.
    \label{eq:base-pgf-normalization}
\end{equation}
Both \(G\) and \(G^{d+1}\) are supported on
\(1+d\mathbb Z_{\geq0}\), since \(d+1\equiv1\pmod d\).
Consequently, \(R\) is a PGF with the required support.

Rearranging its definition gives
\begin{equation}
    G(z)=(1-\gamma)R(z)+\gamma G(z)^{d+1}.
    \label{eq:compound-recursion}
\end{equation}
For \(0\leq z<1\), this identity implies
\begin{align}
\nonumber
    1-G(z)
    &=(1-\gamma)(1-R(z))
      +\gamma(1-G(z)^{d+1})\\
    &\geq
      \gamma(1-G(z))
      \bigl(1+G(z)+\cdots+G(z)^d\bigr).
\label{eq:compound-properness-bound}
\end{align}
Since \(G(z)<1\) for \(z<1\), division by \(1-G(z)\), followed by
\(z\uparrow1\), yields
\begin{equation}
    \gamma(d+1)\leq1.
    \label{eq:compound-gamma-range}
\end{equation}

The zero-constant-term formal solution \(U\) of
\begin{equation}
    U(z)=(1-\gamma)R(z)+\gamma U(z)^{d+1},
    \label{eq:compound-formal-equation}
\end{equation}
is unique.  Indeed, writing \(U(z)=\sum_{n\geq1}u_nz^n\), the coefficient
of \(z^n\) in \(U(z)^{d+1}\) depends only on
\(u_1,\ldots,u_{n-d}\).  Hence \eqref{eq:compound-formal-equation}
determines the coefficients \(u_n\) successively.  Substitution of \(R(z)\)
into \eqref{eq:primitive-gamma-pgf} therefore shows that the solution is
\(F_{d,\gamma}(R(z))\).  Comparing with
\eqref{eq:compound-recursion} proves \eqref{eq:compound-characterization}
and also the uniqueness of \(R\) for fixed \(G\) and \(\gamma\).

Conversely, suppose \(0<\gamma\leq\frac{1}{d+1}\) and \(R\) is any PGF
supported on \(1+d\mathbb Z_{\geq0}\).  By
Proposition~\ref{prop:fuss-catalan-law}, \(F_{d,\gamma}\) is a proper PGF.
Hence \(G=F_{d,\gamma}\circ R\) is a proper PGF with the required support.
Equation \eqref{eq:primitive-gamma-pgf} gives
\eqref{eq:compound-recursion}, so
\(G-\gamma G^{d+1}=(1-\gamma)R\succeq0\).  Finally,
\eqref{eq:compound-fuss-expansion} follows by substituting \(R(z)\) into
the coefficient formula \eqref{eq:fuss-catalan-pmf}.
\end{IEEEproof}

We call the laws in Proposition~\ref{prop:compound-characterization}
\emph{compound weighted Fuss--Catalan repetition laws}.  Their
largest domination constant, for the fixed span \(d\), is an intrinsic
parameter of the law.  For a PGF \(G\) admitting positive domination,
define
\begin{equation}
    \gamma_d^\star(G)
    \coloneq
    \inf_{\{y \colon [z^y]G(z)^{d+1}>0\}}
    \frac{[z^y]G(z)}
         {[z^y]G(z)^{d+1}}.
    \label{eq:optimal-domination-constant}
\end{equation}
The equivalence of the first two conditions in
Proposition~\ref{prop:compound-characterization} shows that
\(\gamma_d^\star(G)\) is the largest admissible constant, and that all
positive admissible constants form the interval
\((0,\gamma_d^\star(G)]\).

The probabilistic representation of a compound law is especially simple.
Let \(N\) have PGF
\(F_{d,\gamma}\), let \(X_1,X_2,\ldots\) be independent with PGF \(R\),
and assume that \(N\) and the \(X_i\)'s are independent.  Then
\begin{equation}
    K=X_1+\cdots+X_N,
    \label{eq:compound-random-sum}
\end{equation}
has PGF \eqref{eq:compound-characterization}.  This random-sum and
branching interpretation is part of the classical Lagrange-distribution
framework \cite{ConsulShenton1972,Good1975,ConsulFamoye2006}; the point
relevant here is that this representation is equivalent to the channel
domination condition.  When \(\gamma<\frac{1}{d+1}\) and \(\E X_1<\infty\),
\begin{equation}
    \E K
    =
    \frac{1-\gamma}{1-(d+1)\gamma}\,\E X_1.
    \label{eq:compound-mean}
\end{equation}
At \(\gamma=\frac{1}{d+1}\), the mean of \(N\), and therefore that of \(K\),
is infinite.

\begin{corollary}[Capacity of the compound class]
\label{cor:compound-capacity}
Every compound law in Proposition~\ref{prop:compound-characterization}
satisfies
\begin{equation}
    \C_{0,q}=\log_2\lambda_{q,d}.
    \label{eq:compound-zero-error-capacity}
\end{equation}
If
\(\gamma_d^\star(G)\geq\lambda_{q,d}^{-d}\), then
\begin{equation}
    \C_q=\C_{0,q}=\log_2\lambda_{q,d}.
    \label{eq:compound-exact-capacity}
\end{equation}
\end{corollary}

\begin{IEEEproof}
Any positive admissible domination constant verifies the hypotheses of
Theorem~\ref{thm:domination}, which gives
\eqref{eq:compound-zero-error-capacity}.  If
\(\gamma_d^\star(G)\geq\lambda_{q,d}^{-d}\), use the maximal constant in
that theorem to obtain \eqref{eq:compound-exact-capacity}.
\end{IEEEproof}

The uniqueness assertion in Proposition~\ref{prop:compound-characterization}
is for fixed \(G\) and \(\gamma\).  A given law can admit several
domination constants, with a different \(R\) for each one.  Corollary
\ref{cor:compound-capacity} uses the intrinsic largest constant, but
\eqref{eq:compound-exact-capacity} remains only a sufficient statement
for the general compound class; unlike the primitive result below, it is
not asserted to be an if-and-only-if threshold.

\subsection{An Explicit Power-Law Example}
\label{sec:fuss-catalan-power}

The characterization above is not limited to laws having Fuss--Catalan
coefficients.  Fix \(\alpha>1\), and consider the repetition law
\begin{equation}
    \Prb(K=1+dk)
    =p_k
    \coloneq
    \frac{1}{\zeta(\alpha)(k+1)^\alpha},
    \qquad k\geq0,
    \label{eq:power-law-pmf}
\end{equation}
where \(\zeta(\cdot)\) is the Riemann zeta function.
Its PGF is
\begin{equation}
    G_\alpha(z)
    =
    \frac{1}{\zeta(\alpha)}
    \sum_{k\geq0}\frac{z^{1+dk}}{(k+1)^\alpha}.
    \label{eq:power-law-pgf}
\end{equation}
Put \(r=d+1\), and let \(p^{*r}\) denote the \(r\)-fold convolution of
\((p_k)_{k\geq0}\).  The coefficient of \(z^{1+dk}\) in
\(G_\alpha(z)^r\) is zero for \(k=0\), while for \(k\geq1\) it equals
\begin{equation}
    [z^{1+dk}]G_\alpha(z)^r=p^{*r}_{k-1}.
    \label{eq:power-law-convolution-coefficient}
\end{equation}
For every tuple of nonnegative integers
\(k_1+\cdots+k_r=k-1\), at least one coordinate satisfies
\begin{equation}
    k_i+1
    \geq
    \frac{k+d}{r}
    \geq
    \frac{k+1}{r}.
    \label{eq:power-law-large-coordinate}
\end{equation}
Taking a union bound over the \(r\) possible coordinates, bounding the
corresponding factor by \(r^\alpha(k+1)^{-\alpha}\), and summing the
remaining \(r-1\) coordinates without the composition constraint gives
\begin{align}
\nonumber
    p^{*r}_{k-1}
    &=
    \frac{1}{\zeta(\alpha)^r}
    \sum_{\substack{k_1+\cdots+k_r=k-1\\
                    k_1,\ldots,k_r\geq0}}
    \prod_{i=1}^r(k_i+1)^{-\alpha}\\
\nonumber
    &\leq
    \frac{r^{\alpha+1}}{\zeta(\alpha)^r(k+1)^\alpha}
    \left(\sum_{j\geq0}(j+1)^{-\alpha}\right)^{r-1}\\
    &=r^{\alpha+1}p_k.
\label{eq:power-law-convolution-bound}
\end{align}
Consequently,
\begin{equation}
    G_\alpha(z)
    \succeq
    \gamma_0G_\alpha(z)^{d+1},
    \qquad
    \gamma_0=(d+1)^{-(\alpha+1)}.
    \label{eq:power-law-domination}
\end{equation}
Proposition~\ref{prop:compound-characterization} therefore supplies a
compound representation of this law for \(\gamma=\gamma_0\), even though
its probabilities are given directly by the power law
\eqref{eq:power-law-pmf}.

The estimate \eqref{eq:power-law-domination} is not claimed to be optimal.
By Proposition~\ref{prop:compound-characterization}, its optimal
domination constant can equivalently be written as
\begin{equation}
    \gamma_d^\star(G_\alpha)
    =
    \inf_{k\geq1}\frac{p_k}{p^{*(d+1)}_{k-1}}
    \geq
    (d+1)^{-(\alpha+1)},
    \label{eq:power-law-optimal-constant}
\end{equation}
but a simpler closed form is not needed here.  The certified value already
implies
\begin{equation}
    \C_{0,q}=\log_2\lambda_{q,d},
    \label{eq:power-law-zero-error-capacity}
\end{equation}
for every \(\alpha>1\), and
\begin{equation}
    \C_q=\log_2\lambda_{q,d}
    \quad\text{whenever}\quad
    \lambda_{q,d}^{d}\geq(d+1)^{\alpha+1}.
    \label{eq:power-law-exact-capacity}
\end{equation}
The mean is finite precisely when \(\alpha>2\), in which case
\begin{equation}
    \E K
    =
    1+d\left(
    \frac{\zeta(\alpha-1)}{\zeta(\alpha)}-1
    \right).
    \label{eq:power-law-mean}
\end{equation}
For instance, take \(q=6\), \(d=4\), and \(\alpha=3\).  The left-hand side
of \eqref{eq:lambda-definition} exceeds one at \(\lambda=5\), and hence
\(\lambda_{6,4}>5\).  Therefore,
\begin{equation}
    \gamma_0=5^{-4}>\lambda_{6,4}^{-4},
    \label{eq:power-law-concrete-threshold}
\end{equation}
and this finite-mean power-law channel satisfies
\begin{equation}
    \C_6=\C_{0,6}=\log_2\lambda_{6,4}.
    \label{eq:power-law-concrete-capacity}
\end{equation}

\subsection{Optimality of the Primitive Domination Constant}

We now return to the primitive PGF \(G=G_{d,\beta}\) of
Proposition~\ref{prop:fuss-catalan-law}.  It is the special compound law
in Proposition~\ref{prop:compound-characterization} obtained by taking
\(\gamma=\beta\) and \(R(z)=z\).

\begin{proposition}[Primitive domination and its optimal constant]
\label{prop:fuss-domination}
For every proper law in Proposition~\ref{prop:fuss-catalan-law},
\begin{equation}
    W_m(y)\geq\beta W_{m+d}(y)
    \qquad \forall m,y\geq1.
\label{eq:fuss-domination}
\end{equation}
Furthermore, \(\beta\) is the largest constant that can replace it
uniformly in \(m\) and \(y\).
\end{proposition}
\begin{IEEEproof}
Multiplying \eqref{eq:algebraic-pgf} by \(G(z)^{m-1}\) gives
\begin{equation}
    G(z)^m
    =
    (1-\beta)zG(z)^{m-1}
    +
    \beta G(z)^{m+d}.
    \label{eq:key-pgf-identity}
\end{equation}
The purpose of this multiplication is to express the defining equation in
terms of the convolution powers corresponding to run lengths \(m\) and
\(m+d\).  Indeed, by definition,
\begin{equation}
    W_j(y)=[z^y]G(z)^j,
\end{equation}
because \(G(z)^j\) is the probability-generating function of the sum of
\(j\) independent repetition counts.  Similarly,
\(zG(z)^{m-1}\) is the generating function of one plus the sum of
\(m-1\) independent repetition counts.

Taking the coefficient of \(z^y\) in
\eqref{eq:key-pgf-identity} therefore yields
\begin{equation}
    W_m(y)
    =
    (1-\beta)[z^{y-1}]G(z)^{m-1}
    +
    \beta W_{m+d}(y).
\end{equation}
All coefficients of \(G(z)\) are nonnegative by
Proposition~\ref{prop:fuss-catalan-law}.  Consequently, every coefficient
of \(G(z)^{m-1}\) is nonnegative as well.  More explicitly, if
\begin{equation}
    G(z)=\sum_{r\geq1}p_r z^r,
    \qquad p_r\geq0,
\end{equation}
then, for \(m\geq2\),
\begin{equation}
    [z^{y-1}]G(z)^{m-1}
    =
    \sum_{\substack{r_1+\cdots+r_{m-1}=y-1\\r_i\geq1}}
    p_{r_1}\cdots p_{r_{m-1}}
    \geq0.
\end{equation}
For \(m=1\), the same conclusion follows from \(G(z)^0=1\).
Since \(1-\beta\geq0\), the preceding coefficient identity implies
\begin{equation}
    W_m(y)-\beta W_{m+d}(y)
    =
    (1-\beta)[z^{y-1}]G(z)^{m-1}
    \geq0.
\end{equation}
This proves \eqref{eq:fuss-domination}.

It remains to show that the constant \(\beta\) is optimal.  Setting
\(m=1\) in \eqref{eq:key-pgf-identity} gives
\begin{equation}
    W_1=(1-\beta)\delta_1+\beta W_{d+1}.
    \label{eq:w1-mixture}
\end{equation}
Equivalently, for every \(y\in\N\),
\begin{equation}
    W_1(y)
    =
    (1-\beta)\mathbf{1}_{\{y=1\}}
    +
    \beta W_{d+1}(y).
\end{equation}
Now
\begin{equation}
    \supp(W_{d+1})
    =
    \{d+1,d+1+d,d+1+2d,\ldots\},
\end{equation}
so \(1\notin\supp(W_{d+1})\).  Hence, for every
\(y\in\supp(W_{d+1})\),
\begin{equation}
    W_1(y)=\beta W_{d+1}(y).
\end{equation}
If a constant \(\gamma\) satisfied
\(W_m(y)\geq\gamma W_{m+d}(y)\) uniformly in \(m\) and \(y\), then taking
\(m=1\) and any \(y\in\supp(W_{d+1})\) would give
\begin{equation}
    \beta W_{d+1}(y)
    =
    W_1(y)
    \geq
    \gamma W_{d+1}(y).
\end{equation}
Since \(W_{d+1}(y)>0\) on its support, this implies
\(\gamma\leq\beta\).  Therefore, \(\beta\) is the largest possible
uniform domination constant.
\end{IEEEproof}

\section{Sharp Capacity for the Primitive Fuss--Catalan Family}
\label{sec:exact-capacity}

\begin{theorem}
\label{thm:main}
Let \(q\geq2\), \(d\geq1\), and
\(0<\beta\leq\frac{1}{d+1}\).  For the repetition law
\eqref{eq:algebraic-pgf} (or \eqref{eq:fuss-catalan-pmf}), the zero-error capacity equals
\begin{equation}
    \C_{0,q}(d,\beta)=\log_2\lambda_{q,d}.
    \label{eq:zero-error-exact}
\end{equation}
The Shannon capacity satisfies
\begin{equation}
 \C_q(d,\beta)=\log_2\lambda_{q,d}
 \quad\Longleftrightarrow\quad
 \beta\geq\lambda_{q,d}^{-d}.
\label{eq:sharp-capacity-characterization}
\end{equation}
For every proper parameter,
\begin{equation}
\begin{split}
 \C_q(d,\beta)
 \leq
 \min\left\{\log_2q,\,
 \max\left\{\log_2\lambda_{q,d},\,
 \frac{1}{d}\log_2\frac{1}{\beta}\right\}\right\}.
\end{split}
\label{eq:fuss-below-threshold-bound}
\end{equation}
\end{theorem}

\begin{IEEEproof}
By Proposition~\ref{prop:fuss-domination},
Theorem~\ref{thm:domination} applies with \(\gamma=\beta\).  It gives
\eqref{eq:zero-error-exact}, \eqref{eq:fuss-below-threshold-bound}, and
the equality in \eqref{eq:sharp-capacity-characterization} whenever
\(\beta\geq\lambda_{q,d}^{-d}\).  Proposition~\ref{prop:sharpness} below
proves strict inequality when \(\beta<\lambda_{q,d}^{-d}\), completing the
converse implication.
\end{IEEEproof}

\begin{proposition}[Strict improvement below the threshold]
\label{prop:sharpness}
Let \(0<\beta\leq\frac{1}{d+1}\).  If
\(\beta<\lambda_{q,d}^{-d}\), then
\begin{equation}
    \C_q(d,\beta)>\log_2\lambda_{q,d}.
    \label{eq:strict-improvement}
\end{equation}
\end{proposition}

\begin{IEEEproof}
Write \(\lambda=\lambda_{q,d}\).  By
\eqref{eq:w1-mixture}, one has
\(W_1=\beta W_{d+1}\) on \(\supp(W_{d+1})\), and hence
\begin{equation}
    \D(W_{d+1}\Vert W_1)=\log_2\frac{1}{\beta}.
    \label{eq:primitive-critical-divergence}
\end{equation}
If \(\beta<\lambda^{-d}\), then
\begin{equation}
    \D(W_{d+1}\Vert W_1)>d\log_2\lambda.
    \label{eq:primitive-kl-failure}
\end{equation}
The condition \eqref{eq:within-residue-kl-condition} therefore fails for
\(r=1\) and \(j=1\), and
Proposition~\ref{prop:exact-kl-criterion} gives
\eqref{eq:strict-improvement}.
\end{IEEEproof}

The capacity at the lower threshold is therefore not merely an artifact of a
chosen dual distribution.  Proposition~\ref{prop:fuss-domination} says that
\(\gamma=\beta\) is the best uniform coefficientwise constant, and
the exact KL criterion shows that the capacity itself leaves the zero-error
value at exactly the same point.

\begin{remark}[Critical endpoint]
\textnormal{
\label{rem:critical}
At \(\beta=\frac{1}{d+1}\), every finite input produces a finite output almost
surely, although the expected output length is infinite.  The fixed-block
converse in Lemma~\ref{lem:qary-dual} uses no moment assumption.
Achievability in Theorem~\ref{thm:run-reduction} uses finitely many input run
lengths and a constant-composition code; after finite output quantization it
is an ordinary finite-alphabet DMC argument.  Hence all statements above
remain operationally valid at criticality without invoking a
finite-mean synchronization-channel theorem.  The operational definition
places no constraint on expected output length or decoding delay.
\myqed}
\end{remark}

Since
\begin{equation}
    (q-1)\sum_{r=1}^{d}q^{-r}=1-q^{-d}<1,
\end{equation}
one has \(\lambda_{q,d}<q\).  The upper bound
\eqref{eq:fuss-below-threshold-bound} is therefore nontrivial below the
threshold only when \(\beta>q^{-d}\).  It is a certificate obtained from
the uniform domination argument; it is not asserted to be the optimal
likelihood-ratio bound for the selected mixture.  Richer dual mixtures may
improve it.

\section{Primitive Fuss--Catalan Parameter Regimes and Explicit Examples}
\label{sec:examples}

For the primitive weighted Fuss--Catalan family, the equality interval is
nonempty among proper laws precisely when
\begin{equation}
    \lambda_{q,d}^{d}\geq d+1;
    \label{eq:proper-intersection}
\end{equation}
it contains finite-mean laws precisely when the inequality is strict.
For fixed \(q\), \(\lambda_{q,d}\) is strictly increasing in \(d\): the
function \((q-1)\sum_{r=1}^{d}x^{-r}\) increases with \(d\) and decreases
with \(x\).  This observation gives the complete classification in
Table~\ref{tab:regimes}.  The upper endpoint shown there is critical and has
infinite mean; it should be replaced by a strict inequality when only finite-mean laws
are admitted.

\begin{table}[H]
\centering
\caption{Primitive proper parameters for which
\(\C_q=\C_{0,q}=\log_2\lambda_{q,d}\)}
\label{tab:regimes}
\small
\begin{tabular}{@{}lll@{}}
\toprule
Alphabet and span & Equality parameters & Finite mean?\\
\midrule
\(q=2,\ d=1,2\) & none & --\\
\(q=2,\ d\geq3\) &
\([\lambda_{2,d}^{-d},\,\frac{1}{d+1}]\) & yes, except at \(\beta=\frac{1}{d+1}\)\\
\(q=3,\ d=1\) & \(\{\frac{1}{2}\}\) & no\\
\(q=3,\ d\geq2\) &
\([\lambda_{3,d}^{-d},\,\frac{1}{d+1}]\) & yes, except at \(\beta=\frac{1}{d+1}\)\\
\(q\geq4,\ d\geq1\) &
\([\lambda_{q,d}^{-d},\,\frac{1}{d+1}]\) & yes, except at \(\beta=\frac{1}{d+1}\)\\
\bottomrule
\end{tabular}
\end{table}

When \(q=2\), one has
\(\lambda_{2,1}=1\) and \(\lambda_{2,2}=(1+\sqrt5)/2\), giving
\(\lambda_{2,1}^{1}<2\) and \(\lambda_{2,2}^{2}<3\).
At \(d=3\), \(\lambda_{2,3}^{3}>4\); monotonicity and induction then give
\(\lambda_{2,d}^{d}>d+1\) for all \(d\geq3\).  When \(q=3\),
\(\lambda_{3,1}=2\), so equality in \eqref{eq:proper-intersection} holds
only at criticality for \(d=1\), while strict inequality holds for
\(d\geq2\).  If \(q\geq4\), then
\(\lambda_{q,d}\geq q-1\geq3\), and
\eqref{eq:proper-intersection} is strict for every \(d\).

\subsection{A Catalan Law with \texorpdfstring{\(q=4,d=1\)}{q=4, d=1}}

The simplest finite-mean nonbinary example has \(q=4\), \(d=1\), and
\(\beta=1/3\).  Here \(\lambda_{4,1}=3\), and
\begin{equation}
 \Prb(K=k+1)
 =
 \frac{1}{k+1}\binom{2k}{k}
 \left(\frac{2}{3}\right)^{k+1}
 \left(\frac{1}{3}\right)^k
 \label{eq:q4-d1-pmf}
\end{equation}
for \(k\geq0\).
Thus \(\beta=\lambda_{4,1}^{-1}\), so this example lies exactly at the
equality threshold in Theorem~\ref{thm:main}, with \(\Delta_2=0\).
The mean repetition count is \(\E K=2\), and
\begin{equation}
    \C_4=\log_2 3=1.58496\ldots
    \quad\text{bits per input symbol}.
    \label{eq:q4-d1-capacity}
\end{equation}

\subsection{A Ternary Fuss--Catalan Law}

For \(q=3\) and \(d=2\),
\begin{equation}
    \lambda_{3,2}=1+\sqrt3,\qquad
    \lambda_{3,2}^{-2}
    =1-\frac{\sqrt3}{2}.
    \label{eq:q3-d2-lambda}
\end{equation}
Choosing \(\beta=1/4\) gives
\begin{equation}
\begin{split}
 \Prb(K=2k+1)
 =
 \frac{1}{2k+1}\binom{3k}{k}
 \left(\frac{3}{4}\right)^{2k+1}
 \left(\frac{1}{4}\right)^k
\end{split}
 \label{eq:q3-d2-pmf}
\end{equation}
for \(k\geq0\),
with \(\E K=3\).  Since
\(
 1-\frac{\sqrt3}{2}<\frac{1}{4}<\frac{1}{3}
\),
Theorem~\ref{thm:main} yields
\begin{equation}
    \C_3=\log_2(1+\sqrt3)=1.44998\ldots \quad\text{bits per input symbol}.
    \label{eq:q3-d2-capacity}
\end{equation}

\subsection{The First Binary Equality Example}

For \(q=2\), \(d=3\), and \(\beta=1/5\),
\begin{equation}
 \Prb(K=3k+1)
 =
 \frac{1}{3k+1}\binom{4k}{k}
 \left(\frac45\right)^{3k+1}
 \left(\frac15\right)^k
 \label{eq:binary-d3-pmf}
\end{equation}
for \(k\geq0\).
The mean is \(\E K=4\).  The number \(\lambda_{2,3}\) is the tribonacci
constant,
\begin{equation}
    \lambda_{2,3}=1.83928\ldots,
    \qquad \lambda_{2,3}^{-3}=0.16071\ldots.
\end{equation}
Consequently,
\begin{equation}
    \C_2=\log_2\lambda_{2,3}
    =0.87914\ldots
    \quad\text{bits per input bit}.
    \label{eq:binary-d3-capacity}
\end{equation}

By contrast, for binary \(d=2\),
\(\lambda_{2,2}=\varphi=(1+\sqrt5)/2\) and
\(\varphi^{-2}>1/3\).  The proper interval
\(\beta\leq1/3\) therefore lies wholly below the equality threshold:
\begin{equation}
    \C_2(2,\beta)>\log_2\varphi
    \qquad(0<\beta\leq1/3).
\end{equation}
For \(1/4<\beta\leq1/3\), \eqref{eq:fuss-below-threshold-bound} also gives
the nontrivial sandwich
\begin{equation}
 \log_2\varphi
 <
 \C_2(2,\beta)
 \leq\frac12\log_2\frac{1}{\beta}.
 \label{eq:binary-d2-sandwich}
\end{equation}
At \(\beta=1/3\), the upper bound is
\(\frac12\log_2 3=0.79248\ldots\).

\section{Conclusion}

We identified a general coefficientwise condition that determines the
zero-error capacity of a \(q\)-ary sticky channel and, in an explicit
parameter regime, makes it equal to the Shannon capacity.  More generally,
the exact KL criterion in Proposition~\ref{prop:exact-kl-criterion}
characterizes when the modulo-\(d\) input law is Shannon-capacity
achieving.  In the exact-capacity regime, the coefficientwise condition
converts a relation among convolution powers of the repetition-count
distribution into an exactly tight capacity-per-unit-cost dual
distribution.

For each fixed admissible domination constant
\(0<\gamma\leq \frac{1}{d+1}\), we also characterized all laws satisfying the
condition.  They are precisely the compound weighted Fuss--Catalan laws
\begin{equation}
    G(z)=F_{d,\gamma}(R(z)),
    \label{eq:conclusion-compound}
\end{equation}
where \(R\) is an arbitrary PGF supported on \(1+d\Z_{\geq0}\).  This characterization
includes explicit laws with power-law tails, such as the power-law example
in Section~\ref{sec:fuss-catalan-power}.

For the primitive weighted Fuss--Catalan family
\begin{equation}
    G(z)=(1-\beta)z+\beta G(z)^{d+1},
    \label{eq:conclusion-primitive}
\end{equation}
the uniform domination constant is exactly \(\beta\).  This yields
\begin{equation}
    \C_q=\C_{0,q}=\log_2\lambda_{q,d}
    \label{eq:conclusion-capacity}
\end{equation}
if and only if \(\beta\geq\lambda_{q,d}^{-d}\) within the proper family.
Below the threshold, a run-length perturbation proves strict separation, and
a second explicit dual gives an analytic upper bound.  The fixed-block proof
also covers the proper critical law, despite its infinite mean.

For a nonprimitive compound law, the intrinsic constant
\(\gamma_d^\star(G)\) gives a checkable sufficient condition but is not known to describe the sharp capacity threshold.  Natural open problems are to
compute this constant for important subclasses and to determine when the
exact KL criterion can hold even though
\(\gamma_d^\star(G)<\lambda_{q,d}^{-d}\).
Sharpening the below-threshold dual and determining the capacity throughout
that regime, including the binary primitive \(d=2\) family, also remain
open.

\section*{Acknowledgment}


During the preparation of this work, the author used OpenAI's ChatGPT to support research, organize literature, and improve the language and \LaTeX\ presentation. The author reviewed and edited all content, and takes full responsibility for the content of the article.

\IfFileExists{IEEEtran.cls}{%
  \appendix[Proof of \eqref{eq:M0}]
}{%
  \appendix
  \section*{Appendix: Proof of \eqref{eq:M0}}
}

Write
\begin{equation}
    \ell_i=r_i+d u_i,
    \qquad u_i\geq0,
\end{equation}
and let \(j=\sum_{i=1}^m u_i\).  The reduced word associated with the
signature is
\begin{equation}
    \bar{x}=a_1^{r_1}\cdots a_m^{r_m}.
\end{equation}
Its length is
\begin{equation}
    \sum_{i=1}^m r_i
    =
    \sum_{i=1}^m\ell_i
    -
    d\sum_{i=1}^m u_i
    =
    n-jd,
\end{equation}
and all its run lengths belong to \(\{1,\ldots,d\}\).  This reduced word is
uniquely determined by the signature.  Conversely, every word of length
\(n-jd\) whose run lengths belong to \(\{1,\ldots,d\}\) determines a
signature realized at length \(n\): one may add \(jd\) symbols to its first
run.  Therefore, signatures realized at length \(n\) are in bijection with
the disjoint union of constrained words of lengths
\begin{equation}
    n,n-d,n-2d,\ldots
\end{equation}
that are positive.

Let \(T_0=1\), and for \(k\geq1\) define
\begin{equation}
    T_k
    =
    \sum_{\substack{h\geq1,\ r_1+\cdots+r_h=k\\
                    1\leq r_i\leq d}}
    (q-1)^h.
\end{equation}
Thus \(T_k\) is the weighted number of compositions of \(k\) with parts in
\(\{1,\ldots,d\}\), where a composition with \(h\) parts has weight \((q-1)^h\).
It satisfies
\begin{equation}
    T_k
    =
    (q-1)\sum_{r=1}^{\min\{d,k\}}T_{k-r}.
\end{equation}
For a fixed composition with \(h\) parts, there are
\(
    q (q-1)^{h-1}
\)
compatible run-symbol sequences: the first run symbol can be chosen in
\(q\) ways and each subsequent run symbol in \(q-1\) ways.  Hence the
number of constrained words of length \(k\) is \(\frac{q}{q-1}T_k\).  The preceding
bijection now gives the exact identity
\begin{equation}
    M_0(n)
    =
    \frac{q}{q-1}
    \sum_{j=0}^{\lfloor(n-1)/d\rfloor}T_{n-jd}.
\end{equation}

It remains to determine the exponential growth rate of this quantity.  The
generating function of \(T_k\) is
\begin{equation}
    \sum_{k\geq0}T_kz^k
    =
    \frac{1}{1-(q-1)\sum_{r=1}^d z^r}.
\label{eq:Tk}
\end{equation}
Let \(\rho=\lambda^{-1}\).  By
\eqref{eq:lambda-definition}, the denominator in \eqref{eq:Tk} vanishes
at \(z=\rho\).  On the other hand, for \(|z|<\rho\),
\begin{equation}
    \left|(q-1)\sum_{r=1}^d z^r\right|
    \leq
    (q-1)\sum_{r=1}^d|z|^r
    <
    (q-1)\sum_{r=1}^d\rho^r
    =
    1.
    \label{eq:Tk-radius-bound}
\end{equation}
Thus the radius of convergence is exactly \(\rho=\lambda^{-1}\).
The Cauchy--Hadamard formula therefore gives
\begin{equation}
    \limsup_{k\to\infty}T_k^{1/k}
    =
    \lambda.
\end{equation}
Since \(M_0(n)\geq\frac{q}{q-1}T_n\), it follows that
\begin{equation}
\label{eq:M0lower}
    \limsup_{n\to\infty}
    \frac{1}{n}\log_2M_0(n)
    \geq
    \log_2\lambda.
\end{equation}

For the reverse inequality, we first show by induction that
\begin{equation}
    T_k\leq\lambda^k
    \qquad\forall k\geq0.
\end{equation}
The assertion is immediate for \(k=0\).  If it holds for all smaller
indices, then
\begin{align}
\nonumber
    T_k
    &=
    (q-1)\sum_{r=1}^{\min\{d,k\}}T_{k-r}\\
\nonumber
    &\leq
    (q-1)\sum_{r=1}^{\min\{d,k\}}\lambda^{k-r}\\
    &\leq
    (q-1)\sum_{r=1}^d\lambda^{k-r}
    =
    \lambda^k,
\end{align}
where the last equality follows from \eqref{eq:lambda-definition}.  If
\(\lambda>1\), the exact counting identity consequently gives
\begin{equation}
    M_0(n)
    \leq
    \frac{q}{q-1}
    \sum_{j\geq0}\lambda^{n-jd}
    =
    \frac{q}
    {(q-1)\bigl(1-\lambda^{-d}\bigr)}
    \lambda^n.
\end{equation}
Hence
\begin{equation}
\label{eq:M0upper}
    \limsup_{n\to\infty}
    \frac{1}{n}\log_2M_0(n)
    \leq
    \log_2\lambda.
\end{equation}

The remaining case is \(\lambda=1\).  From
\eqref{eq:lambda-definition}, this occurs only when
\(q-1=d=1\), or equivalently when \((q,d)=(2,1)\).  In this case
\(T_k=1\) for every \(k\), and the exact counting identity becomes
\begin{equation}
    M_0(n)=2n.
\end{equation}
Therefore,
\begin{equation}
\label{eq:M0upper2}
    \limsup_{n\to\infty}
    \frac{1}{n}\log_2M_0(n)
    =
    0
    =
    \log_2\lambda.
\end{equation}

Combining \eqref{eq:M0lower}, \eqref{eq:M0upper}, and \eqref{eq:M0upper2} proves
\begin{equation}
    \C_{0,q}=\limsup_{n\to\infty}
    \frac{1}{n}\log_2M_0(n)=\log_2\lambda.
\end{equation}

\balance

\end{document}